\documentclass[a4paper,USenglish,cleveref, autoref, thm-restate]{lipics-v2021}

\usepackage{amsmath}
\usepackage{amsthm}
\usepackage{graphicx}
\usepackage[colorinlistoftodos]{todonotes}
\usepackage{algorithm}
\usepackage[noend]{algpseudocode}
\usepackage{enumerate}
\usepackage{xcolor}
\usepackage{framed}
\usepackage{bm}
\usepackage{pifont}
\usepackage{multicol}
\usepackage{lipsum}
\usepackage{tikz}
\usetikzlibrary{calc}
\usepackage{wrapfig}
\usepackage{array}
\usepackage{relsize}
\usepackage{comment}
\usepackage{url}
\usepackage[title]{appendix}
\usepackage{caption}
\usepackage{subcaption}
\usepackage{xstring}

\usepackage{enumitem}

\usepackage{cleveref}

\newtheorem{invariant}{Invariant}
\newtheorem{assumption}{Assumption}

\newenvironment{lemclone}[1]{\noindent \textbf{Lemma~\ref{#1}.}\em}{\par}

\algdef{SE}[RECEIVING]{Receiving}{EndReceiving}[1]{\textbf{upon
receiving}\ #1\ \algorithmicdo}{\algorithmicend\ \textbf{}}%
\algtext*{EndReceiving}

\algdef{SE}[UPON]{Upon}{EndUpon}[1]{\textbf{upon
}\ #1\ \algorithmicdo}{\algorithmicend\ \textbf{}}%
\algtext*{EndUpon}

\algdef{SE}[DELIVERY]{Delivery}{EndDelivery}[1]{\textbf{upon
delivery}\ #1\ \algorithmicdo}{\algorithmicend\ \textbf{}}%
\algtext*{EndDelivery}

\algnewcommand\algorithmicswitch{\textbf{switch}}
\algnewcommand\algorithmiccase{\textbf{case}}
\algnewcommand\algorithmicassert{\texttt{assert}}
\algnewcommand\Assert[1]{\State \algorithmicassert(#1)}%
\algdef{SE}[SWITCH]{Switch}{EndSwitch}[1]{\algorithmicswitch\ #1\ \algorithmicdo}{\algorithmicend\ \algorithmicswitch}%
\algdef{SE}[CASE]{Case}{EndCase}[1]{\algorithmiccase\ #1}{\algorithmicend\ \algorithmiccase}%
\algtext*{EndSwitch}%
\algtext*{EndCase}%

\newcommand{\lr}[1]{\langle #1 \rangle}

\newcommand{\ignore}[1]{}

\title{Tame the Wild with Byzantine Linearizability: Reliable Broadcast, Snapshots, and Asset Transfer}

\author{Shir Cohen}{Technion, Israel}{shirco@campus.technion.ac.il}{}{}
\author{Idit Keidar}{Technion, Israel}{idish@ee.technion.ac.il}{}{}

\authorrunning{S. Cohen and I. Keidar} 

\Copyright{Shir Cohen and Idit Keidar}



\nolinenumbers 

\hideLIPIcs  

\EventEditors{John Q. Open and Joan R. Access}
\EventNoEds{2}
\EventLongTitle{42nd Conference on Very Important Topics (CVIT 2016)}
\EventShortTitle{CVIT 2016}
\EventAcronym{CVIT}
\EventYear{2016}
\EventDate{December 24--27, 2016}
\EventLocation{Little Whinging, United Kingdom}
\EventLogo{}
\SeriesVolume{42}
\ArticleNo{23}

\begin{document}

\maketitle

\begin{abstract}
We formalize Byzantine linearizability, a correctness condition that specifies whether a concurrent object with a sequential specification is resilient against Byzantine failures. Using this definition, we systematically study Byzantine-tolerant emulations of various objects from registers. We focus on three useful objects-- reliable broadcast, atomic snapshot, and asset transfer.
We prove that there exist $n$-process $f$-resilient Byzantine linearizable implementations of such objects from registers if and only if $f<\frac{n}{2}$.
\end{abstract}




    
        
        
     




\section{Introduction}

Over the last decade, cryptocurrencies have taken the world by storm. The idea of a decentralized bank, independent of personal motives has gained momentum, and cryptocurrencies like Bitcoin~\cite{nakamoto2019bitcoin}, Ethereum~\cite{wood2014ethereum}, and Diem~\cite{baudet2019state} now play a big part in the world’s economy. At the core of most of these currencies lies the asset transfer problem. In this problem, there are multiple accounts, operated by processes that wish to transfer assets between accounts.
This environment raises the need to tolerate the malicious behavior of processes that wish to sabotage the system.

In this work, we consider the shared memory model that was somewhat neglected in the Byzantine discussion.
We believe that shared memory abstractions, implemented in distributed settings, allow for an intuitive formulation of the services offered by blockchains and similar decentralized tools.
It is well-known that it is possible to implement reliable read-write shared memory registers via message passing even if a fraction of the servers are Byzantine~\cite{abraham2006byzantine, martin2002minimal,rodrigues2003rosebud,jayanti1998fault}.
As a result, as long as the client processes using the service are not malicious, any fault-tolerant object that can be constructed using registers can also be implemented in the presence of Byzantine servers.
However, it is not clear what can be done with such objects when they are used by Byzantine client processes. In this work, we study this question.

In~\Cref{sec:condition} we define \emph{Byzantine linearizability}, a correctness condition applicable to any shared memory object with a sequential specification. 
Byzantine linearizability addresses the usage of reliable shared memory abstractions by potentially Byzantine client processes.
We then systematically study the feasibility of implementing various Byzantine linearizable shared memory objects from registers.

We observe that existing Byzantine fault-tolerant shared memory constructions~\cite{liskov2005byzantine,mostefaoui2017atomic,abraham2006byzantine} in fact implement Byzantine linearizable registers. Such registers are the starting point of our study.
When trying to implement more complex objects (e.g., snapshots and asset transfer) using registers, constructions that work in the crash-failure model no longer work when Byzantine processes are involved, and new algorithms -- or impossibility results -- are needed.

As our first result, we prove in~\Cref{sec:immpossible} that an asset transfer object used by Byzantine client processes does not have a wait-free implementation, even when its API is reduced to support only transfer operations (without reading processes' balances). Furthermore, it cannot be implemented without a majority of correct processes constantly taking steps.
Asset transfer has wait-free implementations from both reliable broadcast~\cite{auvolat2020money} and snapshots~\cite{guerraoui2019consensus} (which we adapt to a Byzantine version) and thus the same lower bound applies to reliable broadcast and snapshots as well.

In~\Cref{sec:bc}, we present a Byzantine linearizable reliable broadcast algorithm with resilience $f<\frac{n}{2}$, proving that, for this object, the resilience bound is tight. To do so, we define a sequential specification of a reliable broadcast object. Briefly, the object exposes broadcast and deliver operations and we require that deliver return messages previously broadcast. We show that a Byzantine linearizable implementation of such an object satisfies the classical (message-passing) definition~\cite{cachin2011introduction}.
Finally, in~\Cref{sec:snapshot} we present a Byzantine linearizable snapshot with the same resilience.
In contrast, previous constructions of Byzantine lattice agreement, which can be directly constructed from a snapshot~\cite{attiya1992efficient}, required $3f+1$ processes to tolerate $f$ failures.

All in all, we establish a tight bound on the resilience of emulations of three useful shared memory objects from registers. On the one hand, we show that it is impossible to obtain wait-free solutions as in the non-Byzantine model, and on the other hand, unlike previous snapshot and lattice agreement algorithms, our solutions do not require $n>3f$.
Taken jointly, our results yield the following theorem:

\begin{theorem}
In the Byzantine shared memory model, there exist $n$-process $f$-resilient Byzantine linearizable implementations of reliable broadcast, snapshot, and asset transfer objects from registers if and only if $f<\frac{n}{2}$.
\end{theorem}

Although the construction of reliable registers in message passing systems requires $n>3f$ servers, our improved resilience applies to client processes, which are normally less reliable than servers, particularly in the so-called \emph{permissioned model} where servers are trusted and clients are ephemeral.

In summary, we make the following contributions:
\begin{itemize}
    \item Formalizing Byzantine linearizability for any object with a sequential specification.
    
    \item Proving that some of the most useful building blocks in distributed computing, such as atomic snapshot and reliable broadcast, do not have $f$-resilient implementations from SWMR registers when $f\geq\frac{n}{2}$ processes are Byzantine.
    
    \item Presenting Byzantine linearizable implementations of a reliable broadcast object and a snapshot object with the optimal resilience.
\end{itemize}

\section{Related Work}
\label{sec:related}

In~\cite{aguilera2019impact} Aguilera et al. present a non-equivocating broadcast algorithm in shared memory. This broadcast primitive is weaker than reliable broadcast -- it does not guarantee that all correct processes deliver the same messages, but rather that they do not deliver conflicting messages.
A newer version of their work~\cite{aguilera2021impact}, developed concurrently and independently of our work\footnote{Their work~\cite{aguilera2021impact} was in fact published shortly after the initial publication of our results~\cite{cohen2021tame}.}, also implements reliable broadcast with $n\geq 2f+1$, which is very similar to our implementation.
While the focus of their work is in the context of RDMA in the \emph{M\&M} (message--and--memory) model, our work focuses on the classical shared memory model, which can be emulated in classical message passing systems. While the algorithms are similar, we formulate reliable broadcast as a shared memory object, with designated API method signatures, which allows us to reason about the operation interval as needed for proving (Byzantine) linearizability and for using this object in constructions of other shared memory objects.

Given a reliable broadcast object, there are known implementations of lattice agreement~\cite{di2020byzantine,DBLP:conf/opodis/ZhengG20}, which resembles a snapshot object.
However, these constructions require $n=3f+1$ processes. 
In our work, we present both Byzantine linearizable reliable broadcast and Byzantine snapshot, (from which Byzantine lattice agreement can be constructed~\cite{attiya1992efficient}), with resilience $n=2f+1$.

The asset transfer object we discuss in this paper was introduced by Guerraoui et al.~\cite{guerraoui2019consensus,collins2020online}. 
Their work provides a formalization of the cryptocurrency definition~\cite{nakamoto2019bitcoin}.
The highlight of their work is the observation that the asset transfer problem can be solved without consensus. It is enough to maintain a partial order of transactions in the systems, and in particular, every process can record its own transactions.
They present a wait-free linearizable implementation of asset transfer in crash-failure shared memory, taking advantage of an atomic snapshot object. We show that we can use their solution, together with our Byzantine snapshot, to solve Byzantine linearizable asset transfer with $n=2f+1$.

In addition, Guerraoui et al. present a Byzantine-tolerant solution in the message passing model. This algorithm utilizes reliable broadcast, where dependencies of transactions are explicitly broadcast along with the transactions. This solution does not translate to a Byzantine linearizable one, but rather to a sequentially consistent asset transfer object. In particular, reads can return old (superseded) values, and transfers may fail due to outdated balance reads.

Finally, recent work by Auvolat et al.~\cite{auvolat2020money} continues this line of work.
They show that a FIFO order property between each pair of processes is sufficient in order to solve the asset transfer problem. This is because transfer operations can be executed once a process's balance becomes sufficient to perform a transaction and there is no need to wait for all causally preceding transactions.
However, as a result, their algorithm is not sequentially consistent, or even causally consistent for that matter. For example, assume process $i$ maintains an invariant that its balance is always at least 10, and performs a transfer with amount 5 after another process deposits 5 into its account, increasing its balance to 15. Using the protocol in~\cite{auvolat2020money}, another process might observe $i$'s balance as 5 if it sees $i$'s outgoing transfer before the causally preceding deposit.
Because our solution is Byzantine linearizable, such anomalies are prevented.
\section{Model and Preliminaries}
\label{sec:model}

We study a distributed system in the shared memory model.
Our system consists of a well-known static set $\Pi=\{1,\dots,n\}$ of asynchronous client processes.
These processes have access to some shared memory objects.
In the shared memory model, all communication between processes is done through the API exposed by the objects in the system: processes invoke operations that in turn, return some response to the process.
In this work, we assume a reliable shared memory. (Previous works have presented constructions of such reliable shared memory in the message passing model~\cite{abraham2006byzantine, martin2002minimal,rodrigues2003rosebud,afek1995computing,jayanti1998fault}).
We further assume an adversary that may adaptively corrupt up to $f$ processes in the course of a run.
When the adversary corrupts a process, it is defined as \emph{Byzantine} and may deviate arbitrarily from the protocol. As long as a process is not corrupted by the adversary, it is \emph{correct}, follows the protocol, and takes infinitely many steps. In particular, it continues to invoke the object's API infinitely often. Later in the paper, we show that the latter assumption is necessary.

We enrich the model with a \emph{public key infrastructure} (PKI). That is, every process is equipped with a public-private key pair used to sign data and verify signatures of other processes. We denote a value $v$ signed by process $i$ as $\lr{v}_i$.


{\bf Executions and Histories.}
We discuss algorithms emulating some object $O$ from lower level objects (e.g., registers). An  algorithm is organized as methods of $O$. A method execution is a sequence of \emph{steps}, beginning with the method’s invocation (invoke step), proceeding through steps that access lower level objects (e.g., register read/write), and ending with a return step. The invocation and response delineate the method’s execution interval.
In an \emph{execution} $\sigma$ of a Byzantine shared memory algorithm, each correct process invokes methods sequentially, where steps of different processes are interleaved.
Byzantine processes take arbitrary steps regardless of the protocol.
The \emph{history} $H$ of an execution $\sigma$ is the sequence of high-level invocation and response events of the emulated object $O$ in $\sigma$.

A \emph{sub-history} of a history $H$ is a sub-sequence of the events of $H$.
A history $H$ is \emph{sequential} if it begins with an invocation and each invocation, except possibly the last, is immediately followed by a matching response.
Operation $op$ is pending in a history $H$ if $op$ is invoked in $H$ but does not have a matching response event.

A history defines a partial order on operations: operation $op_1$ precedes $op_2$ in history $H$, denoted $op_1\prec _H op_2$, if the response event of $op_1$ precedes the invocation event of $op_2$ in $H$. Two operations are concurrent if neither precedes the other.

{\bf Linearizability.}
A popular correctness condition for concurrent objects in the crash-fault model is linearizability~\cite{herlihy1990linearizability}, which is defined with respect to an object's sequential specification.
A \emph{linearization} of a concurrent history
$H$ of object $o$ is a sequential history $H'$
such that (1) after removing some pending operations from $H$ and completing others by adding matching responses, it contains the same invocations and responses as $H'$, (2) $H'$ preserves the partial order $\prec_H$, and (3) $H'$ satisfies $o$'s sequential specification.

{\bf f-resilient.}
An algorithm is \emph{f-resilient} if as long as at most $f$ processes fail, every correct process eventually returns from each operation it invokes. A \emph{wait-free} algorithm is a special case where $f=n-1$.

{\bf Single Writer Multiple Readers Register.}
The basic building block in shared memory is a single writer multiple readers (SWMR) register that exposes \emph{read} and \emph{write} operations. Such registers are used to construct more complicated objects.
The sequential specification of a SWMR register states that every read operation from register $R$ returns the value last written to $R$.
Note that if the writer is Byzantine, it can cause a correct reader to read arbitrary values.

{\bf Asset Transfer Object.}
In~\cite{guerraoui2019consensus,collins2020online}, the asset transfer problem is formulated as a sequential object type, called \emph{Asset Transfer Object}. The asset transfer object maintains a mapping from processes in the system to their balances\footnote{The definition in~\cite{guerraoui2019consensus} allows processes to own multiple accounts. For simplicity, we assume a single account per-process, as in \cite{collins2020online}.}. Initially, the mapping contains the initial balances of all processes. The object exposes a \emph{transfer} operation, \emph{transfer(src,dst,amount)}, which can be invoked by process \emph{src} (only). It withdraws \emph{amount} from process \emph{src}'s account and deposits it at process \emph{dst}'s account provided that \emph{src}'s balance was at least \emph{amount}. It returns a boolean that states whether the transfer was successful (i.e., \emph{src} had \emph{amount} to spend).
In addition, the object exposes a \emph{read(i)} operation that returns the current balance of $i$.

\section{Byzantine Linearizability}
\label{sec:condition}
In this section we define Byzantine linearizability.
Intuitively, we would like to tame the Byzantine behavior in a way that provides consistency to correct processes.
We linearize the correct processes' operations and offer a degree of freedom to embed additional operations by Byzantine processes.

We denote by $H|_{correct}$ the projection of a history $H$ to all correct processes.
We say that a history $H$ is Byzantine linearizable if $H|_{correct}$ can be augmented with operations of Byzantine processes such that the completed history is linearizable. That is, there is another history, with the same operations by correct processes as in $H$, and additional operations by another at most $f$ processes.
In particular, if there are no Byzantine failures then Byzantine linearizability is simply linearizability.
Formally:

\begin{definition}{(Byzantine Linearizability)}
\label{def:byzlin}
A history $H$ is Byzantine linearizable
if there exists a history $H'$ so that $H'|_{correct}=H|_{correct}$ and $H'$ is linearizable.

\end{definition}

Similarly to linearizability, we say that an object is Byzantine linearizable if all of its executions are Byzantine Linearizable.

Next, we characterize objects for which Byzantine linearizability is meaningful. The most fundamental component in shared memory is read-write registers. Not surprisingly, such registers, whether they are single-writer or multi-writers ones are de facto Byzantine linearizable without any changes. This is because before every read from a Byzantine register, invoked by a correct process, one can add a corresponding Byzantine write.

In practice, multiple writers multiple readers (MWMR) registers are useless in a Byzantine environment as an adversary that controls the scheduler can prevent any communication between correct processes. 
SWMR registers, however, are still useful for constructing more meaningful objects. Nevertheless, the constructions used in the crash-failure model for linearizable objects do not preserve this property.
For instance, if we allow Byzantine processes to run a classic atomic snapshot algorithm~\cite{afek1993atomic} using Byzantine linearizable SWMR registers, it will not result in a Byzantine linearizable snapshot object. 
The reason is that the algorithm relies on correct processes being able to perform ``double-collect'' meaning that at some point a correct process manages to read all registers twice without witnessing any changes. While this is true in the crash-failure model, in the Byzantine model this is not the case as the adversary can change some registers just before any correct read.

\paragraph*{Relationship to Other Correctness Conditions}
Byzantine linearizability provides a simple and intuitive way to capture Byzantine behavior in the shared memory model.
We now examine the relationship of Byzantine linearizability with previously suggested correctness conditions involving Byzantine processes.

PBFT~\cite{castro1999practical,castro1999correctness} presented a formalization of linearizability in the presence of Byzantine-faulty clients in message passing systems. Their notion of linearizability is formulated in the form of I/O automata. Their specification is in the same spirit as ours, but our formulation is closer to the original notion of linearizability in shared memory.

Some works have defined linearization conditions for specific objects.
This includes conditions for SWMR registers~\cite{mostefaoui2017atomic}, a distributed ledger~\cite{cholvi2020atomic}, and asset transfer~\cite{auvolat2020money}.
Our condition coincides with these definitions for the specific objects and thus generalizes all of them.
Liskov and Rodrigues~\cite{liskov2005byzantine} presented a correctness condition that has additional restrictions. 
Their correctness notion relies on the idea that Byzantine processes are eventually detected and removed from the system and focuses on converging to correct system behavior after their departure. While this model is a good fit when the threat model is software bugs or malicious intrusions, it is less appropriate for settings like cryptocurrencies, where Byzantine behavior cannot be expected to eventually stop.

\section{Lower Bound on Resilience}
\label{sec:immpossible}

In shared memory, one typically aims for wait-free objects, which tolerate any number of process failures. Indeed, many useful objects have wait-free implementations from SWMR registers in the non-Byzantine case. This includes reliable broadcast, snapshots, and as recently shown, also asset transfer. We now show that in the Byzantine case, wait-free implementations of these objects are impossible. Moreover, a majority of correct processes is required.

\begin{theorem}
\label{theorem:majority_needed}
In the Byzantine shared memory model, for any $f>2$, there does not exist a Byzantine linearizable implementation of asset transfer that supports only transfer operations in a system with $n\leq 2f$ processes, $f$ of which can be Byzantine, using only SWMR registers.
\end{theorem}

Note that to prove this impossibility, it does not suffice to introduce bogus actions by Byzantine processes, because the notion of Byzantine linearizability allows us to ignore these actions. Rather, to derive the contradiction, we create runs where the bogus behavior of the Byzantine processes leads to incorrect behavior of the correct processes.

\begin{proof}
Assume by contradiction that there is such an algorithm. Let us look at a system with $n=2f$ correct processes.
Partition $\Pi$ as follows: $\Pi=A\cup B\cup \{p_1,p_2\}$, where $|A|=f-1$, $|B|=f-1$, $A\cap B=\emptyset$, and $p_1,p_2\notin A\cup B$. By assumption, $|A|>1$. Let $z$ be a process in $A$. Also, by assumption $|B|\geq 2$. Let $q_1,q_2$ be processes in $B$.
The initial balance of all processes but $z$ is $0$, and the initial balance of $z$ is $1$.
We construct four executions as shown in~\Cref{fig:impossibility}.

Let $\sigma_1$ be an execution where, only processes in $A\cup\{p_1\}$ take steps. First, $z$ performs \emph{transfer($z,p_1$,1)}.
Since up to $f$ processes may be faulty, the operation completes, and by the object's sequential specification, it is successful (returns true). 
Then, $p_1$ performs \emph{transfer($p_1,q_1$,1)}. 
By $f$-resilience and linearizability, this operation also completes successfully.
Note that in $\sigma_1$ no process is actually faulty, but because of $f$-resilience, progress is achieved when $f$ processes are silent.

Similarly, let $\sigma_2$ be an execution where the processes in $A\cup\{p_2\}$ are correct, and $z$ performs \emph{transfer($z,p_2$,1)}, followed by \emph{$p_2$} performing \emph{transfer($p_2,q_2$,1)}.

We now construct $\sigma_3$, where all processes in $A\cup\{p_1\}$ are Byzantine. 
We first run $\sigma_1$. Call the time when it ends $t_1$. At this point, all processes in $A\cup\{p_1\}$ restore their registers to their initial states. Note that no other processes took steps during $\sigma_1$, hence the entire shared memory is now in its initial state.
Then, we execute $\sigma_2$. Because we have reset the memory to its initial state, the operations execute the same way.  
When $\sigma_2$ completes, processes in \emph{A}$\setminus\{z\}\cup\{p_1\}$ restore their registers to their state at time $t_1$. 
At this point, the state of $z$ and $p_2$ is the same as it was at the end of $\sigma_2$, the state of processes in $A\setminus\{z\}\cup\{p_1\}$ is the same as it was at the end of $\sigma_1$, and processes in $B$ are all in their initial states.

We construct $\sigma_4$ where all processes in $A\cup\{p_2\}$ are Byzantine
by executing $\sigma_2$, having all processes in $A\cup\{p_2\}$ reset their memory, executing $\sigma_1$, and then having $z$ and $p_2$ restore their registers to their state at the end of $\sigma_2$.
At this point, the state of $z$ and $p_2$ is the same as it was at the end of $\sigma_2$, the state of processes in $A\setminus\{z\}\cup\{p_1\}$ is the same as it was at the end of $\sigma_1$, and processes in $B$ are all in their initial states.

We observe that for processes in \emph{B}, the configurations at the end of $\sigma_3$ and $\sigma_4$ are indistinguishable as they did not take any steps and the global memory is the same.
By $f$-resilience, in both cases $q_1$ and $q_2$, together with processes in \emph{B} and one of $\{p_1,p_2\}$ should be able to make progress at the end of each of these runs.
We extend the runs by having $q_1$ and $q_2$ invoke transfers of amount 1 to each other.
In both runs processes in $B\cup\{p_1,p_2\}$ help them make progress. In $\sigma_3$, $p_1$ behaves as if it is a correct process and its local state is the same as it is at the end of $\sigma_1$, and in $\sigma_4$ $p_2$ behaves as if it is a correct process and its local state is the same as it is at the end of $\sigma_2$. Thus, $\sigma_3$ and $\sigma_4$ are indistinguishable to all correct processes, and as a result $q_1$ and $q_2$ act the same in both runs. 
However, from safety exactly one of their transfers should succeed.
In $\sigma_3$, $p_2$ is correct and
\emph{transfer($p_2,q_2$,1)} succeeds, allowing $q_2$ to transfer 1 and disallowing the transfer from $q_1$, whereas $\sigma_4$ the opposite is true.
This is a contradiction.
\end{proof}

\begin{figure}
    \centering
    \includegraphics[scale=0.75]{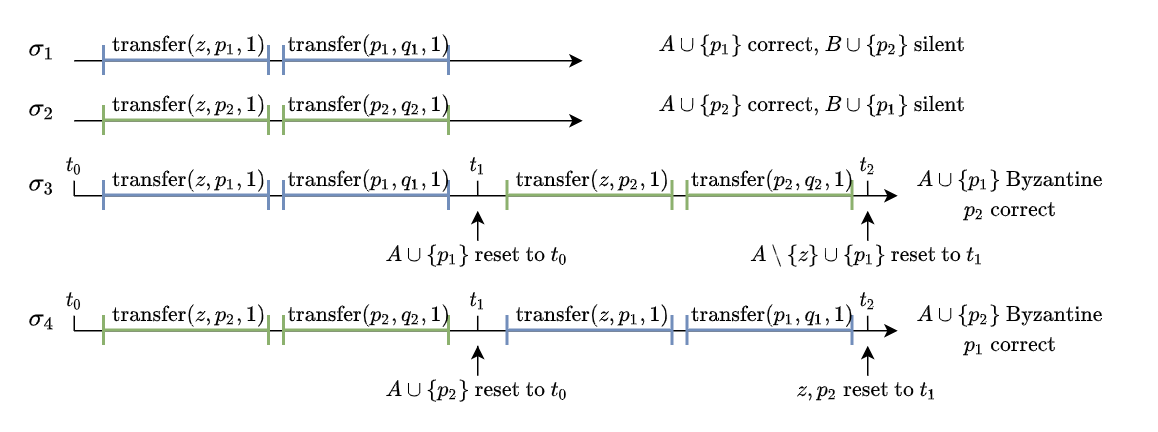}
    \caption{An asset transfer object does not have an $f$-resilient implementation for $n\leq 2f$.}
    \label{fig:impossibility}
\end{figure}

Guerraoui et al.~\cite{guerraoui2019consensus} use an atomic snapshot to implement an asset transfer object in the crash-fault shared memory model. 
In addition, they handle Byzantine processes in the message passing model by taking advantage of reliable broadcast. In~\Cref{at_appendix} we show that their atomic snapshot-based asset transfer can be easily adapted to the Byzantine settings by using a Byzantine linearizable snapshot, resulting in a Byzantine linearizable asset transfer.
Their reliable broadcast-based algorithm, on the other hand, is not linearizable and therefore not Byzantine linearizable even when using Byzantine linearizable reliable broadcast.
Nonetheless, Auvolat et al.~\cite{auvolat2020money} have used reliable broadcast to construct an asset transfer object where transfer operations are linearizable (although reads are not).

We note that our lower bound holds for an asset transfer object without read operations.
This discussion and the construction in~\Cref{at_appendix} lead us to the following corollary:

\begin{corollary}
In the Byzantine shared memory model, for any $f>2$, there does not exist an $f$-resilient Byzantine linearizable implementation of an atomic snapshot or reliable broadcast in a system with $f\geq\frac{n}{2}$ Byzantine processes using only SWMR registers.
\end{corollary}


Furthermore, we prove in the following lemma that in order to provide $f$-resilience it is required that at least a majority of correct processes take steps infinitely often, justifying our model definition.

\begin{lemma}
\label{cor:infoften}
In the Byzantine shared memory model, for any $f>2$, there does not exist an $f$-resilient Byzantine linearizable implementation of asset transfer in a system with $n\geq 2f+1$ processes, $f$ of which can be Byzantine, using only SWMR registers if less than $f+1$ correct processes take steps infinitely often.
\end{lemma}
\begin{proof}
Assume by way of contradiction that there exists an $f$-resilient Byzantine linearizable implementation of asset transfer in a system with $n\geq 2f+1$ processes where there are at most $f$ correct processes that take steps infinitely often. Denote these $f$ correct processes by the set $A$.
Thus, there is a point $t$ in any execution such that from time $t$, only processes in $A$ and Byzantine processes take any steps.
Starting $t$, the implementation is equivalent to one in a system with $n=2f$, $f$ of them may be Byzantine. This is a contradiction to~\Cref{theorem:majority_needed}.
\end{proof}

\section{Byzantine Linearizable Reliable Broadcast}
\label{sec:bc}

With the acknowledgment that not all is possible, we seek to find Byzantine linearizable objects that are useful even without a wait-free implementation.
One of the practical objects is a reliable broadcast object. We already proved in the previous section that it does not have an $f$-resilient Byzantine linearizable implementation, for any $f\geq max\{3, \frac{n}{2}\}$. In this section we provide an implementation that tolerates $f<\frac{n}{2}$ faults.

\subsection{Reliable Broadcast Object}

The reliable broadcast primitive exposes two operations \emph{broadcast(ts,m)} returning void and \emph{deliver(j,ts)} returning $m$.
When $deliver_j(i, ts)$ returns $m$ we say that process $j$ delivers $m$ from process $i$ in timestamp $ts$. 
The broadcast operation allows processes to spread a message $m$ in the system, along with some timestamp $ts$.
The use of timestamps allows processes to broadcast multiple messages.

Its classical definition, given for message passing systems~\cite{cachin2011introduction}, requires the following properties:
\begin{itemize}
    \item Validity: If a correct process $i$ broadcasts $(ts,m)$ then all correct processes eventually deliver $m$ from process $i$ in timestamp $ts$.
    
    \item Agreement: If a correct process delivers $m$ from process $i$ in timestamp $ts$, then all correct processes eventually deliver $m$ from process $i$ in timestamp $ts$.

    \item Integrity: No process delivers two different messages for the same $(ts,j)$ and if $j$ is correct delivers only messages $j$ previously broadcast.

\end{itemize}

In the shared memory model, 
the deliver operation for some process $j$ and timestamp $ts$ returns the message with timestamp $ts$ previously broadcast by $j$, if exists. We define the sequential specification of reliable broadcast as follows:

\begin{definition}
\label{def:bc}
A reliable broadcast object exposes two operations \emph{broadcast(ts,m)} returning void and \emph{deliver(j,ts)} returning m.
A call to \emph{deliver(j,ts)} returns the value m of the first \emph{broadcast(ts,m)} invoked by process $j$ before the deliver operation.
If $j$ did not invoke \emph{broadcast} before the deliver, then it returns $\bot$.
\end{definition}

Note that as the definition above refers to sequential histories, the first broadcast operation (if such exists) is well-defined. Further,
whereas in message passing systems reliable broadcast works in a push fashion, where the receipt of a message triggers action at its destination, in the shared memory model processes need to actively pull information from the registers.
A process pulls from another process $j$ using the \emph{deliver(j,ts)} operation and returns with a value $m\neq\bot$.
If all messages are eventually pulled, the reliable broadcast properties are achieved, as proven in the following lemma.

\begin{lemma}
\label{app:rb_prop}
A Byzantine linearization of a reliable broadcast object satisfies the three properties of reliable broadcast.
\end{lemma}
\begin{proof}
If a correct process broadcasts $m$, and all messages are subsequently pulled then according to~\Cref{def:bc} all correct processes deliver $m$, providing validity. For agreement, if a correct process invokes \emph{deliver(j,ts)} that returns $m$ and all messages are later pulled by all correct processes, it follows that all correct processes also invoke \emph{deliver(j,ts)} and eventually return $m'\neq \bot$.
Since \emph{deliver(j,ts)} returns the value $v$ of the first \emph{broadcast(ts,v)} invoked by process $j$ before it is called, and there is only one first broadcast, and we get that $m=m'$.
Lastly, if \emph{deliver(j,ts)} returns $m$, by the specification, $j$ previously invoked \emph{broadcast(ts,m)}.

\end{proof}

\subsection{Reliable Broadcast Algorithm}

In our implementation (given in~\Cref{alg:bc}), each process has 4 SWMR registers: send, echo, ready, and deliver, to which we refer as \emph{stages} of the broadcast. We follow concepts from Bracha's implementation in the message passing model~\cite{bracha1987asynchronous} but leverage the shared memory to improve its resilience from $3f+1$ to $2f+1$. The basic idea is that a process that wishes to broadcast value $v$ writes it in its send register (line \ref{l:add_to_send}) and returns only when it reaches the deliver stage. I.e., $v$ appears in the deliver register of at least one correct process.
Throughout the run, processes infinitely often call a \emph{refresh} function whose role is to help the progress of the system. When refreshing, processes read all registers and help promote broadcast values through the 4 stages. For a value to be delivered, it has to have been read and signed by $f+1$ processes at the ready stage. Because each broadcast message is copied to 4 registers of each process, the space complexity is $4n$ per message. Whether this complexity can be improved remains as an open question.

In the refresh function, executed for all processes, at first a process reads the last value written to a send register (line~\ref{l:get_val}).
If the value is a signed pair of a message and a timestamp, refresh then copies it to the process’s echo register in line~\ref{l:add_to_echo}.
In the echo register, the value remains as evidence, preventing conflicting values (sent by Byzantine processes) from being delivered. That is,
before promoting a value to the ready or deliver stage, a correct process $i$ performs a ``double-collect'' of the echo registers (in lines~\ref{line:conf1},\ref{line:conf2}). Namely, after collecting $f+1$ signatures on a value in ready registers, meaning that it was previously written in the echo of at least one correct process, $i$ re-reads all echo registers to verify that there does not exist a conflicting value (with the same timestamp and sender). Using this method, concurrent deliver operations ``see'' each other, and delivery of conflicting values broadcast by a Byzantine process is prevented.
Before delivering a value, a process writes it to its deliver register with $f+1$ signatures (line~\ref{line:add_to_deliver}). Once one correct process delivers a value, the following deliver calls can witness the $f+1$ signatures and copy this value directly from its deliver register (line~\ref{line:copy_deliver}).

\def\NoNumber#1{{\def\alglinenumber##1{}\State #1}\addtocounter{ALG@line}{-1}}

\begin{algorithm}[htb]

    \caption{Shared Memory Bracha: code for process $i$}
    \label{alg:bc}

    \begin{algorithmic}[1]
    
    \Statex shared SWMR registers: $send_i, echo_i, ready_i, deliver_i$
    \Statex

    \Procedure{conflicting-echo}{$\lr{ts,v}_j$}
        \State return $\exists w\neq v, k\in \Pi$ such that $\lr{ts,w}_j\in echo_k$
    \EndProcedure
    \Statex
    
    \Procedure{broadcast}{ts,val}

        \State $send_i\gets \lr{ts,val}_i$ \label{l:add_to_send}

        \Repeat
            \State $m\gets$ deliver(i,ts)\label{line:call_deliver}
            \Until $m\neq\bot$\Comment{message is deliverable}
    
    \EndProcedure
    \Statex        
   
    \Procedure{deliver}{j,ts}
        \State refresh()
        \If{$\exists k\in \Pi$ and $v$ s.t. $\lr{\lr{ts,v}_j,\sigma}\in deliver_k$ where $\sigma$ is a set of $f+1$ signatures on $\lr{ready,\lr{ts,v}_j}$} \label{line:complete_deliver}
        \State $deliver_i\gets deliver_i\cup \{\lr{\lr{ts,v}_j,\sigma} \}$\label{line:copy_deliver}
        \State return $v$
        \EndIf
        \State return $\bot$

    \EndProcedure
    \Statex

    \Procedure{refresh}{}

        \For{$j\in[n]$}

            \State $m\gets send_j$ \label{l:get_val}
            \If{$\nexists ts$, $val$ s.t. $m=\lr{ts,val}_j$}
                continue \Comment{$m$ is not a signed pair}
            \EndIf
            \State $echo_i\gets echo_i\cup \{m\}$ \label{l:add_to_echo}

            \If{$\neg$conflicting-echo$(m)$}\label{line:conf1}
                \State $ready_i\gets ready_i\cup \{\lr{ready, m}_i\}$
            \EndIf
            
            \If{$\exists S\subseteq \Pi$ s.t. $|S|\geq f+1$, $\forall j\in S, \lr{ready,m}_j\in ready_j$ and $\neg$conflicting-echo$(m)$} \label{line:conf2}

                \State $deliver_i\gets deliver_i\cup \{\lr{m,\sigma=\{\lr{ready,m}_j|j\in S\}}\}$ \Comment{$\sigma$ is the set of $f+1$ signatures}\label{line:add_to_deliver}
            \EndIf

    \EndFor

    \EndProcedure

    \end{algorithmic}

\end{algorithm}

We make two assumptions on the correct usage of our algorithm. The first is inherently required as shown in~\Cref{cor:infoften}:
\begin{assumption}
\label{assumption:refresh}
All correct processes infinitely often invoke methods of the reliable broadcast API.
\end{assumption}

The second is a straight forward validity assumption:

\begin{assumption}
\label{assumption:onets}
Correct processes do not invoke \emph{broadcast(ts,val)} twice with the same $ts$.
\end{assumption}

In~\Cref{bc_appendix} we prove the correctness of the reliable broadcast algorithm and conclude the following theorem:

\begin{theorem}
\Cref{alg:bc} implements an $f$-resilient Byzantine linearizable reliable broadcast object for any $f<\frac{n}{2}$.
\end{theorem}

\section{Byzantine Linearizable Snapshot}
\label{sec:snapshot}

In this section, we utilize a reliable broadcast primitive to construct a Byzantine snapshot object with resilience $n>2f$.

\subsection{Snapshot Object}
A snapshot~\cite{afek1993atomic} is represented as an array of $n$ shared single-writer variables that can be accessed with two
operations: \emph{update(v)}, called by process $i$, updates the $i^{th}$ entry in the array and \emph{snapshot} returns an array.
The sequential specification of an atomic snapshot is as follows: the $i^{th}$ entry of the array returned by a \emph{snapshot} invocation contains the value \emph{v} last updated by an \emph{update(v)} invoked by process $i$, or its variable's initial value if no update was invoked.

Following~\Cref{cor:infoften}, we again must require that correct processes perform operations infinitely often. 
For simplicity, we require that they invoke infinitely many snapshot operations; if processes invoke either snapshots or updates, we can have each update perform a snapshot and ignore its result.

\begin{assumption}
\label{assumption:inf_snaps}
All correct processes invoke snapshot operations infinitely often.
\end{assumption}

\def\NoNumber#1{{\def\alglinenumber##1{}\State #1}\addtocounter{ALG@line}{-1}}

\begin{algorithm}[htb]

    \begin{algorithmic}[1]
    
    \Statex shared SWMR registers:  $\forall j\in[n]$ $collected_i[j]\in \{\bot\}\cup\{\mathbb{N}\times Vals\}$ with selectors $ts$ and val, initially $\bot$
    \Statex $\forall k\in\mathbb{N}$, $savesnap_i[k]\in \{\bot\}\cup\{\text{array of }n\; Vals\times \text{set of messages}\}$ with selectors $snap$ and proof, initially $\bot$
    
    \Statex local variables: $ts_i\in \mathbb{N}$, initially 0
    \Statex $\forall j\in[n]$, $rts_i[j]\in \mathbb{N}$, initially 0
    \Statex $r,auxnum\in \mathbb{N}$, initially 0
    \Statex $p\in [n]$, initially 1
    \Statex  $\forall j\in[n],k\in\mathbb{N}$, \emph{seen$_i$[j][k]},$senders_i\in\mathcal{P}(\Pi)$, initially $\emptyset$ 
    \Statex $\sigma\gets \emptyset$ set of messages
    \Statex

    \Procedure{update}{$v$}
         \For {$j\in[n]$}\Comment{collect current memory state}
            \State update-collect($collected_j$) \label{line:collect_in_up}
        \EndFor
        \State $ts_i\gets ts_i+1$ \label{line:inc_ts}
        \State $collected_i[i]\gets \lr{ts_i,v}_i$\label{line:update_collect1} \Comment{update local component of collected}

    \EndProcedure
    \Statex

       \Procedure{snapshot}{}
       \For {$j\in[n]$}\Comment{collect current memory state}
            \State update-collect($collected_j$)\label{line:do_collect}
        \EndFor
        \State $c\gets collected_i$ \label{line:updatedc}
        \Repeat \label{line:s_r}
            \State $auxnum\gets auxnum+1$
            \State $snap\gets$ snapshot-aux$(auxnum)$ \label{line:ret_aux}
        \Until $snap\geq c$ \label{line:f_r} \label{line:return_snapshot}\Comment{snapshot is newer than the collected state}
        \State return $snap$
    
    \EndProcedure

    \Statex
   
    \Procedure{update-collect}{c}
        \For {$k\in[n]$} \label{line:s_u}
            \If{$c[k].ts>collected_i[k].ts$ and $c[k]$ is signed by $k$}\label{line:mono_check} 
                \State $collected_i[k]\gets c[k]$\label{line:update_collect2}
            \EndIf
        \EndFor \label{line:f_u}
    \EndProcedure

    \algstore{part1}
    \end{algorithmic}

    \caption{Byzantine Snapshot: code for process $i$}
    \label{alg:snap1}

\end{algorithm}

\begin{algorithm}

    \begin{algorithmic}[1]
   \algrestore{part1}

    \Procedure{minimum-saved}{auxnum}
    
    \State $S\gets \{s| \exists j\in[n], s=savesnap_j[auxnum].snap$ and $savesnap_j[auxnum].proof$ is a valid proof of $s\}$
    \If{$S=\emptyset$} \label{line:checkstored1}
    
            \State return $\bot$
    \EndIf
        \State $res\gets$ infimum$(S)$ \label{line:set_ret} \Comment{returns the minimum value in each index}
        \State $savesnap_i[auxnum]\gets \lr{res,\bigcup_{j\in[n]} savesnap_j[auxnum].proof}$ \label{line:storesnap1}
        
        \State update-collect($res$) \label{line:up_collect_w_save}

        \State return $res$ \label{line:ret_inf}
    
    \EndProcedure
    \Statex
    
    \Procedure{snapshot-aux}{auxnum}
        \State initiate new reliable broadcast instance \label{line:new_bc}
        \State $\sigma\gets\emptyset$
       \For {$j\in[n]$}\Comment{collect current memory state}\label{line:do_col2}
            \State update-collect($collected_j$)\label{line:do_collect2}
        \EndFor
        \State $senders_i\gets\{i\}$ \Comment{start message contains collect} \label{line:add_to_starts_local}
        \State broadcast(0,$\lr{collect_i}_i$) \label{line:round_0_bc}

        \While{true}

            \State $cached\gets$ minimum-saved($auxnum$) \label{line:check_cache}  \Comment{check if there is a saved snapshot}     \If{$cached\neq\bot$}
                return $cached$ \label{line:ret_arr1}
            \EndIf
            
            \State $p\gets (p+1)$ mod $n +1$ \Comment{deliver messages in round robin}\label{bcmec3}
      
            \State $m\gets$ deliver($p,rts_i[p]$) \label{bcmec1} \Comment{deliver next message from $p$}
            
            \If{$m = \bot$} \label{line:ignore_bot}
                 continue
            \EndIf
            
                
            \If{$rts_i[p]$= 0 and $m$ contains a signed collect array $c$}
            \Statex \Comment{start message (round 0)}
                \State $\sigma\gets \sigma\cup\{m\}$
                \State update-collect($c$) \label{line:invoke_update_collect}
                \State $senders_i\gets senders_i\cup\{j\} $ \label{line:up_senders_1}

        \ElsIf{$m$ contains a signed set of processes, $jsenders$} 
        \Statex \Comment{round $r$ message for $r>0$}

        \If {$jsenders\nsubseteq senders_i$} \label{line:check_contains}
    
            \State continue \Comment{cannot process message, its dependencies are missing}
   
        \EndIf
        \State $\sigma\gets \sigma\cup\{m\}$

       \State $seen_i[j][rts_i[p]]\gets jsenders\cup seen_i[j][rts_i[p]-1]$   \label{line:updateseen}

        \EndIf
    \State $rts_i[p]\gets rts_i[p]+1$\label{bcmec2}

        \Statex

    \If {received $f+1$ round-$r$ messages for the first time} \label{line:beginarounds}
            \State $r\gets r+1$ \label{line:next_round}
            \State broadcast($r,\lr{senders_i}_i$)\label{line:bc_rounds}
    \EndIf

    \Statex
  
      \If{$\exists s$ s.t.
        $|\{j |$ $seen_i[j][s]=senders_i\}|=f+1$} \Comment{stability condition} \label{line:requestedcondition}

            \State $r\gets 0$
            \State  \emph{senders$_i$}$\gets\emptyset$ 

            \State  $\forall j\in[n],k\in\mathbb{N}$, \emph{seen$_i$[j][k]}$\gets \emptyset$ 

            
            \State $cached\gets$ minimum-saved($auxnum$) \Comment{re-check for saved snapshot}
            \If{$cached\neq\bot$}
            return $cached$ \label{line:ret_arr2}
            \EndIf
            
            \State $savesnap_i[auxnum]\gets \lr{collect_i,\sigma}$ 
            \Statex \Comment{$\sigma$ contains all received messages in this snapshot-aux instance} \label{line:updatesavesnap}\label{line:storesnap2}

            \State return $collect_i$ \label{line:return_snapaux}

        \EndIf

        \EndWhile
    
    \EndProcedure

    \end{algorithmic}

    \caption{Byzantine Snapshot auxiliary procedures: code for process $i$}
    \label{alg:snap2}

\end{algorithm}

\subsection{Snapshot Algorithm}

Our pseudo-code is presented in~\Cref{alg:snap1,alg:snap2}.
During the algorithm, we compare snapshots using the (partial) coordinate-wise order. That is, let $s_1$ and $s_2$ be two $n$-arrays. We say that $s_2>s_1$ if $\forall i\in[n]$, $s_2[i].ts>s_1[i].ts$.

Recall that all processes invoke snapshot operations infinitely often.
In each snapshot instance, correct processes start by collecting values from all registers and broadcasting their collected arrays in ``start'' messages (message with timestamp 0).
Then, they repeatedly send the identities of processes from which they delivered start messages until there exists a round such that the same set of senders is received from $f+1$ processes in that round. Once this occurs, it means that the $f+1$ processes see the exact same start messages and the snapshot is formed as the supremum of the collects in their start messages.

We achieve optimal resilience by waiting for only $f+1$ processes to send the same set. Although there is not necessarily a correct process in the intersection of two sets of size $f+1$, we leverage the fact that reliable broadcast prevents equivocation to ensure that nevertheless, there is a common \emph{message} in the intersection, so two snapshots obtained in the same round are necessarily identical. Moreover, once one process obtains a snapshot $s$, any snapshot seen in a later round exceeds $s$.

Each process $i$ collects values from all processes' registers in a shared variable $collect_i$.
When starting a snapshot operation, each process runs update-collect, where it updates its collect array (line~\ref{line:do_collect}) and saves it in a local variable $c$ (line~\ref{line:updatedc}).
When it does so, it updates the $i^{th}$ entry to be the highest-timestamped value it observes in the $i^{th}$ entries of all processes' collect arrays (lines \ref{line:s_u} --~\ref{line:f_u}).
Then, it initiates the snapshot-aux procedure with a new auxnum tag.
Snapshot-aux returns a snapshot, but not necessarily a ``fresh'' one that reflects all updates that occurred before \emph{snapshot} was invoked. Therefore, snapshot-aux is repeatedly called until it collects a snapshot $s$ such that $s\geq c$, according to the snapshots partial order (lines~\ref{line:s_r} --~\ref{line:f_r}).

By~\Cref{assumption:inf_snaps} and since the $auxnum$ variable at each correct process is increased by 1 every time snapshot-aux is called, all correct processes participate in all instances of snapshot-aux. 
When a correct process invokes a snapshot-aux procedure with auxnum, it first initiates a new reliable broadcast instance at line~\ref{line:new_bc}, dedicated to this instance of snapshot-aux.
Note that although processes invoke one snapshot-aux at a time, they may engage in multiple reliable broadcast instances simultaneously. That is, they continue to partake in previous reliable broadcast instances after starting a new one.
As another preliminary step of snapshot-aux, each correct process once again updates its collect array using the update-collect procedure (lines~\ref{line:do_col2}--~\ref{line:do_collect2}) and broadcasts it to all processes at line~\ref{line:round_0_bc}.
During the execution, a correct processes delivers messages from all other processes in a round robin fashion. The local variable $p$ represents the process from which it currently delivers. In addition, $rts[p]$
maintains the next timestamp to be delivered from $p$ (lines~\ref{bcmec1}, \ref{bcmec2}, \ref{bcmec3}). Note that if the delivered message at some point is $\bot$, $rts[p]$ is not increased, so all of $p$'s messages are delivered in order (line~\ref{line:ignore_bot}).

Snapshot-aux proceeds in rounds, which are reflected in the timestamps of the messages broadcast during its execution. Each correct process starts snapshot-aux at round $0$, where it broadcasts its collected array; we refer to this as its start message.
It then continues to round $r+1$ once it has delivered $f+1$ round $r$ messages (line~\ref{line:next_round}).
Each process maintains a local set $senders$ that contains the processes from which it received start messages (line~\ref{line:up_senders_1}).
In every round (from 1 onward) processes send the set of processes from which they received start messages (line~\ref{line:bc_rounds}).

Process $i$ maintains a local map $seen[j][r]$ that maps a process $j$ and a round $r$ to the set of processes that $j$ reported to have received start messages from in rounds 1--r (line~\ref{line:updateseen}), but only if $i$ has received start messages from all the reported processes (line~\ref{line:check_contains}).
By doing so, we ensure that if for some correct process $i$ and a round r $seen_i[j][r]$ contains a process $l$, $l$ is also in $senders_i$.
If this condition is not satisfied, the delivered counter for $j$ ($rts[j]$) is not increased and this message will be repeatedly delivered until the condition is satisfied.

Once there is a process $i$ such that there exists a round $s$ and there is a set $S$ of $f+1$ processes $j$ for which $seen_i[j][s]$ is equal to $senders_i$, we say that the \emph{stability condition} at line~\ref{line:requestedcondition} is satisfied for $S$. At that time, $i$ and $f$ more processes agree on the collected arrays sent at round 0 by processes in $senders_i$, and $collect_i$ holds the supremum of those collected arrays.
This is because whenever it received a start message, it updated its collect so that currently $collect_i$ reflects all collects sent by processes in $senders_i$.
Thus, $i$ can return its current collect as the snapshot-aux result.
Since reliable broadcast prevents Byzantine processes from equivocating, there are $f$ more processes that broadcast the same $senders$ set at that round, and any future round will ``see'' this set.
As we later show, after at most $n+1$ rounds, the stability condition holds and hence the size of $seen$ is $O(n^3)$. Together with the collected arrays, the total space complexity is cubic in $n$.

To ensure liveness in case some correct processes complete a snapshot-aux instance before all do, we add a helping mechanism.
Whenever a correct process successfully completes snapshot-aux, it stores its result in a savesnap map, with the auxnum as the key (either at line~\ref{line:storesnap1} or at line~\ref{line:storesnap2}).
This way, once one correct process returns from snapshot-aux, others can read its result at line~\ref{line:check_cache} and return as well.
To prevent Byzantine processes from storing invalid snapshots, each entry in the savesnap map is a tuple of the returned array and a proof of the array's validity. The proof is the set of messages received by the process that stores its array in the current instance of snapshot-aux. Using these messages, correct processes can verify the legitimacy of the stored array. If a correct process reads from savesnap a tuple with an invalid proof, it simply ignores it.

\subsection{Correctness}

We outline the key correctness arguments highlighting the main lemmas. Formal proofs of these lemmas appear in~\Cref{snapshot_appendix}.
To prove our algorithm is Byzantine linearizable, we first show that all returned snapshots are totally ordered (by coordinate-wise order):

\begin{lemma}
\label{lemma:concurrent_snaps}
If two snapshot operations invoked by correct processes return $s_i$ and $s_j$, then $s_j\geq s_i$ or $s_j<s_i$. 
\end{lemma}

Based on this order, we define a linearization. Then, we show that our linearization preserves real-time order, and it respects the sequential specification.
We construct the linearization $E$ as follows: First, we linearize all snapshot operations of correct processes in the order of their return values. Then, we linearize every update operation by a correct process immediately before the first snapshot operation that ``sees'' it.  We say that a snapshot returning $s$ \emph{sees} an update by process $j$ that has timestamp $ts$ if $s[j].ts\geq ts$. If multiple updates are linearized to the same point (before the same snapshot), we order them by their start times.
Finally, we add updates by Byzantine processes as follows: We add \emph{update(v)} by a Byzantine process $j$ if there is a linearized snapshot that returns $s$ and $s[j].val=v$. We add the update immediately before any snapshot that sees it.

We next prove that the linearization respects the sequential specification. 

\begin{lemma}
\label{lemma:seq}
The $i^{th}$ entry of the array returned by a \emph{snapshot} invocation contains the value $v$ last updated by an \emph{update(v)} invoked by process $i$ in $E$, or its variable's initial value if no update was invoked.
\end{lemma}

Because an update is linearized immediately before some snapshot sees it and snapshots are monotonically increasing, all following snapshots see the update as well. Next, we prove in the two following lemmas that $E$ preserves the real-time order.

\begin{lemma}
\label{lemma:following_snaps}
If a snapshot operation invoked by a correct process $i$ with return value $s_i$ precedes a snapshot operation invoked by a correct process $j$ with return value $s_j$, then $s_i\leq s_j$. 
\end{lemma}

\begin{lemma}
\label{lemma:snap_sees_update}
Let $s$ be the return value of a snapshot operation $snap_i$ invoked by a correct process $i$. Let $update_j(v)$ be an update operation invoked by a correct process $j$ that writes $\lr{ts,v}$ and completes before $snap_i$ starts. Then, $s[j].ts\geq ts$.
\end{lemma}

It follows from~\Cref{lemma:snap_sees_update} and the definition of $E$, that if an update precedes a snapshot it is linearized before it, and from~\Cref{lemma:following_snaps} that if a snapshot precedes a snapshot it is also linearized before it.
The following lemma ensures that if an update precedes another update it is linearized before it. That is, if a snapshot operation sees the second update, it sees the first one.

\begin{lemma}
\label{sees}
If update1 by process $i$ precedes update2 by process $j$ and a snapshot operation $snap$ by a correct process sees update2, then $snap$ sees update1 as well.
\end{lemma}

Finally, the next lemma proves the liveness of our algorithm.

\begin{lemma}{(Liveness)}
\label{liveness}
Every correct process that invokes some operation eventually returns.
\end{lemma}

We conclude the following theorem:

\begin{theorem}
\Cref{alg:snap1} implements an $f$-resilient Byzantine linearizable snapshot object for any $f<\frac{n}{2}$.
\end{theorem}
\begin{proof}
\Cref{lemma:concurrent_snaps} shows that there is a total order on snapshot operations. Using this order, we have defined a linearization $E$ that satisfies the sequential specification (\Cref{lemma:seq}). We then proved that $E$ also preserves real-time order (Lemmas~\ref{lemma:following_snaps} --~\ref{sees}). Thus, \Cref{alg:snap1} is Byzantine linearizable. In addition,~\Cref{liveness} proves that \Cref{alg:snap1} is $f$-resilient.
\end{proof}

\section{Conclusions}
\label{sec:conclusions}

We have studied shared memory constructions in the presence of Byzantine processes.
To this end, we have defined Byzantine linearizability, a correctness condition suitable for shared memory algorithms that can tolerate Byzantine behavior. We then used this notion to present both upper and lower bounds on some of the most fundamental components in distributed computing. 

We proved that atomic snapshot, reliable broadcast, and asset transfer are all problems that do not have $f$-resilient emulations from registers when $n\leq2f$. On the other hand, we have presented an algorithm for Byzantine linearizable reliable broadcast with resilience $n>2f$. We then used it to implement a Byzantine snapshot with the same resilience. Among other applications, this Byzantine snapshot can be utilized to provide a Byzantine linearizable asset transfer. Thus, we proved a tight bound on the resilience of emulations of asset transfer, snapshot, and reliable broadcast. 

Our paper deals with feasibility results and does not focus on complexity measures. In particular, we assume unbounded storage in our constructions.
We leave the subject of efficiency as an open question for future work.

\bibliographystyle{plainurl}
\bibliography{references}

\begin{thebibliography}{10}

\bibitem{abraham2006byzantine}
Ittai Abraham, Gregory Chockler, Idit Keidar, and Dahlia Malkhi.
\newblock Byzantine disk paxos: optimal resilience with byzantine shared
  memory.
\newblock {\em Distributed Computing}, 18(5):387--408, 2006.

\bibitem{afek1993atomic}
Yehuda Afek, Hagit Attiya, Danny Dolev, Eli Gafni, Michael Merritt, and Nir
  Shavit.
\newblock Atomic snapshots of shared memory.
\newblock {\em Journal of the ACM (JACM)}, 40(4):873--890, 1993.

\bibitem{afek1995computing}
Yehuda Afek, David~S Greenberg, Michael Merritt, and Gadi Taubenfeld.
\newblock Computing with faulty shared objects.
\newblock {\em Journal of the ACM (JACM)}, 42(6):1231--1274, 1995.

\bibitem{aguilera2019impact}
Marcos~K Aguilera, Naama Ben-David, Rachid Guerraoui, Virendra Marathe, and
  Igor Zablotchi.
\newblock The impact of rdma on agreement.
\newblock In {\em Proceedings of the 2019 ACM Symposium on Principles of
  Distributed Computing}, pages 409--418, 2019.

\bibitem{aguilera2021impact}
Marcos~K. Aguilera, Naama Ben-David, Rachid Guerraoui, Virendra Marathe, and
  Igor Zablotchi.
\newblock The impact of rdma on agreement, 2021.
\newblock \href {http://arxiv.org/abs/1905.12143} {\path{arXiv:1905.12143}}.

\bibitem{attiya1992efficient}
Hagit Attiya, Maurice Herlihy, and Ophir Rachman.
\newblock Efficient atomic snapshots using lattice agreement.
\newblock In {\em International Workshop on Distributed Algorithms}, pages
  35--53. Springer, 1992.

\bibitem{cohen2021tame}
Anonymous Author(s).
\newblock Tame the wild with byzantine linearizability: Reliable broadcast,
  snapshots, and asset transfer, arxiv (reference omitted for blind review), 21
  Feb 2021.

\bibitem{auvolat2020money}
Alex Auvolat, Davide Frey, Michel Raynal, and Fran{\c{c}}ois Ta{\"\i}ani.
\newblock Money transfer made simple: a specification, a generic algorithm, and
  its proof.
\newblock {\em Bulletin of EATCS}, 3(132), 2020.

\bibitem{baudet2019state}
Mathieu Baudet, Avery Ching, Andrey Chursin, George Danezis, Fran{\c{c}}ois
  Garillot, Zekun Li, Dahlia Malkhi, Oded Naor, Dmitri Perelman, and Alberto
  Sonnino.
\newblock State machine replication in the libra blockchain.
\newblock {\em The Libra Assn., Tech. Rep}, 2019.

\bibitem{bracha1987asynchronous}
Gabriel Bracha.
\newblock Asynchronous byzantine agreement protocols.
\newblock {\em Information and Computation}, 75(2):130--143, 1987.

\bibitem{cachin2011introduction}
Christian Cachin, Rachid Guerraoui, and Lu{\'\i}s Rodrigues.
\newblock {\em Introduction to reliable and secure distributed programming}.
\newblock Springer Science \& Business Media, 2011.

\bibitem{castro1999correctness}
Miguel Castro, Barbara Liskov, et~al.
\newblock A correctness proof for a practical byzantine-fault-tolerant
  replication algorithm.
\newblock Technical report, Technical Memo MIT/LCS/TM-590, MIT Laboratory for
  Computer Science, 1999.

\bibitem{castro1999practical}
Miguel Castro, Barbara Liskov, et~al.
\newblock Practical byzantine fault tolerance.
\newblock In {\em OSDI}, volume~99, pages 173--186, 1999.

\bibitem{cholvi2020atomic}
Vicent Cholvi, Antonio~Fernandez Anta, Chryssis Georgiou, Nicolas Nicolaou, and
  Michel Raynal.
\newblock Atomic appends in asynchronous byzantine distributed ledgers.
\newblock In {\em 2020 16th European Dependable Computing Conference (EDCC)},
  pages 77--84. IEEE, 2020.

\bibitem{collins2020online}
Daniel Collins, Rachid Guerraoui, Jovan Komatovic, Petr Kuznetsov, Matteo
  Monti, Matej Pavlovic, Yvonne-Anne Pignolet, Dragos-Adrian Seredinschi,
  Andrei Tonkikh, and Athanasios Xygkis.
\newblock Online payments by merely broadcasting messages.
\newblock In {\em 2020 50th Annual IEEE/IFIP International Conference on
  Dependable Systems and Networks (DSN)}, pages 26--38. IEEE, 2020.

\bibitem{di2020byzantine}
Giuseppe~Antonio Di~Luna, Emmanuelle Anceaume, and Leonardo Querzoni.
\newblock Byzantine generalized lattice agreement.
\newblock In {\em 2020 IEEE International Parallel and Distributed Processing
  Symposium (IPDPS)}, pages 674--683. IEEE, 2020.

\bibitem{guerraoui2019consensus}
Rachid Guerraoui, Petr Kuznetsov, Matteo Monti, Matej Pavlovi{\v{c}}, and
  Dragos-Adrian Seredinschi.
\newblock The consensus number of a cryptocurrency.
\newblock In {\em Proceedings of the 2019 ACM Symposium on Principles of
  Distributed Computing}, pages 307--316, 2019.

\bibitem{herlihy1990linearizability}
Maurice~P Herlihy and Jeannette~M Wing.
\newblock Linearizability: A correctness condition for concurrent objects.
\newblock {\em ACM Transactions on Programming Languages and Systems (TOPLAS)},
  12(3):463--492, 1990.

\bibitem{jayanti1998fault}
Prasad Jayanti, Tushar~Deepak Chandra, and Sam Toueg.
\newblock Fault-tolerant wait-free shared objects.
\newblock {\em Journal of the ACM (JACM)}, 45(3):451--500, 1998.

\bibitem{liskov2005byzantine}
Barbara Liskov and Rodrigo Rodrigues.
\newblock Byzantine clients rendered harmless.
\newblock In {\em International Symposium on Distributed Computing}, pages
  487--489. Springer, 2005.

\bibitem{martin2002minimal}
Jean-Philippe Martin, Lorenzo Alvisi, and Michael Dahlin.
\newblock Minimal byzantine storage.
\newblock In {\em International Symposium on Distributed Computing}, pages
  311--325. Springer, 2002.

\bibitem{mostefaoui2017atomic}
Achour Most{\'e}faoui, Matoula Petrolia, Michel Raynal, and Claude Jard.
\newblock Atomic read/write memory in signature-free byzantine asynchronous
  message-passing systems.
\newblock {\em Theory of Computing Systems}, 60(4):677--694, 2017.

\bibitem{nakamoto2019bitcoin}
Satoshi Nakamoto.
\newblock Bitcoin: A peer-to-peer electronic cash system.
\newblock Technical report, Manubot, 2009.

\bibitem{rodrigues2003rosebud}
Rodrigo Rodrigues and Barbara Liskov.
\newblock Rosebud: A scalable byzantine-fault-tolerant storage architecture.
\newblock Technical report, 2003.

\bibitem{wood2014ethereum}
Gavin Wood et~al.
\newblock Ethereum: A secure decentralised generalised transaction ledger.
\newblock {\em Ethereum project yellow paper}, 151(2014):1--32, 2014.

\bibitem{DBLP:conf/opodis/ZhengG20}
Xiong Zheng and Vijay~K. Garg.
\newblock Byzantine lattice agreement in asynchronous systems.
\newblock In Quentin Bramas, Rotem Oshman, and Paolo Romano, editors, {\em 24th
  International Conference on Principles of Distributed Systems, {OPODIS} 2020,
  December 14-16, 2020, Strasbourg, France (Virtual Conference)}, volume 184 of
  {\em LIPIcs}, pages 4:1--4:16. Schloss Dagstuhl - Leibniz-Zentrum f{\"{u}}r
  Informatik, 2020.
\newblock \href {https://doi.org/10.4230/LIPIcs.OPODIS.2020.4}
  {\path{doi:10.4230/LIPIcs.OPODIS.2020.4}}.

\end{thebibliography}

\appendix
 \newpage
\begin{appendices}
\section{Byzantine Asset Transfer}
\label{at_appendix}

\def\NoNumber#1{{\def\alglinenumber##1{}\State #1}\addtocounter{ALG@line}{-1}}

\begin{algorithm}

    \begin{algorithmic}[1]
    
    \Statex shared Byzantine snapshot: $S$
    \Statex initial-- immutable array of initial balances
    \Statex local variables: $txns_i$ -- sets of outgoing transaction, initially $\{\}$ 
    \Statex $ts_i\in \mathbb{N}$, initially 0

    \Statex $snap$ -- array of sets of transactions, initially array of empty sets \Comment{the last snapshot taken}
    \Statex
    
   \Statex \textbf{struct} $txn$ \textbf{contains}:
    
    timestamp ts,
    
    source src,
    
    destination dst,
    
    amount amount
    \Statex

    \Procedure{balance}{j,snap}
    
        \State $incoming\gets0$
        \State $outgoing\gets0$
        \For{$l\in [n]$}
            \For{$k\in snap[l]$}
                \If {$snap[l][k].dst=j$ and valid($snap[l][k]$)}
                    \State $incoming\gets incoming+snap[l][k].amount$
                \EndIf
            \EndFor
        \EndFor
        
        \For{$k\in snap[j]$}
            \If {valid($snap[j][k]$)}
                \State $outgoing\gets outgoing+snap[j][k].amount$
            \EndIf
        \EndFor
        
        \State return $initial(j)+incoming-outgoing$
     
    \EndProcedure
    \Statex

    \Procedure{transfer}{src,dst,amount}
        \State $ts_i\gets ts_i+1$
        \State $snap\gets S.snapshot()$ \label{l:s_op2}
        \If{$balance(src,snap)<amount$} \label{line:check_bal}
            \State return false
            
        \EndIf
        \State $txns_i\gets txns_i.append( \lr{ts_i,src,dst,amount,snap}_i)$
        \State $S.update(txns_i)$ \label{line:lin_of_up}
        \State return true

    \EndProcedure
    \Statex

    \Procedure{read}{j}

        \State $snap\gets S.snapshot()$ \label{l:s_op1}
        \State return $balance(j,snap)$

    \EndProcedure

    \end{algorithmic}

    \caption{Byzantine Asset Transfer: code for process $i$}
    \label{alg:at}

\end{algorithm}

In this section we adapt the asset transfer implementation from snapshots given in~\cite{guerraoui2019consensus} to a Byzantine asset transfer.
The algorithm is very simple. It is based on a shared snapshot array $S$, with a cell for each client process $i$, representing $i$’s outgoing transactions. An additional immutable array holds all processes’ initial balances. A process $i$’s balance is computed by taking a snapshot of $S$ and applying all of $i$’s valid incoming and outgoing transfers to $i$’s initial balance. A transfer invoked by process $i$ checks if $i$’s balance is sufficient, and if so, appends the transfer details (source, destination, and amount) to $i$’s cell.
Similarly to the use of dependencies in the (message-passing broadcast-based) asset transfer algorithm of [17], we also track the history of every transaction.
To this end, we append to the process’s cell also the snapshot taken to compute the balance for each transaction.

\begin{theorem}
\Cref{alg:at} implements an $f$-resilient Byzantine linearizable asset transfer object for any $f<\frac{n}{2}$.
\end{theorem}
\begin{proof}

At any point during a sequential execution, we denote by $B(p)$ the balance of process $p$.
Recall that the operation $transfer(src,dst,amount)$ causes the following changes: $B(src)=B(src)-amount$ and $B(dst)=B(dst)+amount$.

In addition, at any point during a concurrent execution, we represent by $balance(p)$ the balance of process $p$ derived from the state as follows:

If $p$ is a correct process:
\begin{equation*}
  \begin{aligned}
    balance(p) & \overset{def}{=} initial(p)\cr
      & + \sum_{j\in correct(\Pi)}{amount\:|\: txn=\lr{*,j,p,amount,*}\in S[j] \wedge  valid(txn)}\cr
      & + \sum_{j\in Byzantine(\Pi)}{amount\:|\: txn=\lr{*,j,p,amount,*}\in S[j] \wedge  valid(txn)}  \cr
      & \wedge \text{txn was read by some correct process}\cr
      & - \sum_{j\in \Pi}{amount\:|\: txn=\lr{*,p,j,amount,*}\in S[p] \wedge  valid(txn)}
  \end{aligned}
\end{equation*}

If $p$ is a Byzantine process:
\begin{equation*}
  \begin{aligned}
    balance(p) & \overset{def}{=} initial(p)\cr
      & + \sum_{j\in correct(\Pi)}{amount\:|\: txn=\lr{*,j,p,amount,*}\in S[j] \wedge  valid(txn)}\cr
        & + \sum_{j\in Byzantine(\Pi)}{amount\:|\: txn=\lr{*,j,p,amount,*}\in S[j] \wedge  valid(txn)}  \cr
        & \wedge \text{txn was read by some correct process}\cr
      & - \sum_{j\in \Pi}{amount\:|\: txn=\lr{*,p,j,amount,*}\in S[p] \wedge  valid(txn)} \cr
      & \wedge \text{txn was read by some correct process}
  \end{aligned}
\end{equation*}

Let us examine an execution $E$ of the algorithm.
Let $H$ be the history of $E$.
First, we define $H^c$ to be the history $H$ after removing any pending read operations and any pending transfer operations that did not complete line~\ref{line:lin_of_up}.
We define $H'$ to be an augmentation of $H^c|{correct}$ as follows.

For every Byzantine process $j$ and a transaction $txn=(ts,j,dst,amount,deps)$ such that $txn$ appears in the array returned by the snapshot procedure (either in line~\ref{l:s_op1} or line~\ref{l:s_op2}) for at least one correct process $i$, we add to $H'$ a \emph{transfer$_j$(j,dst,amount)} operation that begins and ends immediately before the first correct process performs that snapshot procedure.
Since at least one correct process reads this transaction, this moment is well-defined.
We construct a linearization $E'$ of $H'$ by defining the following linearization points:

\begin{itemize}
 
    \item Let $o$ be a \emph{read$_i$(j)} operation by a correct process $i$ that completes line~\ref{l:s_op2}. The operation linearizes at that moment.
    
    \item Let $o$ be a \emph{transfer$_i$(i,dst,amount)} operation by a correct process $i$ that completed line~\ref{line:lin_of_up}. The operation linearizes at that moment.
    Note that operations that do not complete this line are removed from $H'$. By the code, these lines are between the invocation and the return of the broadcast procedure.

    \item If $o$ is a completed \emph{transfer$_i$(i,dst,amount)} operation by a correct process $i$ that returns false it linearizes at line~\ref{l:s_op1}.

    \item Every Byzantine \emph{transfer$_j$(j,dst,amount)} operation by process $j$ linearizes at the moment we added it.
\end{itemize}

In $H'$ there are no read operations by Byzantine processes. It is clear from construction that each operation invoked by a correct process is mapped to some point between its invocation event and its response event.
We now prove that the concrete concurrent run simulates the specification. That is, if we execute the sequential run defined by the linearization points the changes in the balances (represented by B) reflects the actual changes on $balance$. 
Before the execution begins, $B(p)$ is the initial balance of process $p$. As the snapshot is empty before the run begins, it holds by definition that $B(p)=balance(p)$. We now show that at any point $B(p)=balance(p)$.

We prove the equivalence of $B(p)$ and $balance(p)$ by induction on the steps in the executions. 
We assume that the claim holds before a particular step and show that it remains the same after each step. For a correct process $p$, $balances(p)$ changes at line~\ref{line:lin_of_up} when some transfer involving $p$ is updated in the snapshot. As this is the linearization point of a transfer operation, the same change in balance also applies to $B(p)$ at that moment. For a Byzantine process $p$, $balances(p)$ changes at line~\ref{l:s_op2} or line~\ref{l:s_op2} when its transaction is being read by a correct process. A transfer operation by Byzantine processes is added immediately before the first correct process reads it, so this change also reflect $B(p)$ at that moment.

Next, we prove $f$-resilience.

\begin{lemma}{(Liveness)}
Every correct process that invokes some operation eventually returns.
\end{lemma}
\begin{proof}
This is immediate from the snapshot $f$-resilient guarantees and the fact that all other operations are local computations.
\end{proof}

\end{proof}

\section{Reliable Broadcast: Correctness}
\label{bc_appendix}

We now prove our reliable broadcast algorithm's correctness. We first notice:

\begin{observation}
  If process $i$ is correct and $v$ appears in $echo_i$ or $ready_i$ it is never deleted.\label{inv:remains} 
\end{observation}

\begin{lemma}
 If process $i$ is correct and $\lr{\lr{ts,v}_i,\sigma}$ appears in $deliver_j$ for any process $j$ then $i$ previously invoked $broadcast(ts, v)$.\label{inv:bc_occured}
\end{lemma}
\begin{proof}
 Since we assume unforgeable signatures, $i$ has previously signed $\lr{ts,v}$. By the code, this is only possible if $i$ invoked $broadcast(ts,v)$.
\end{proof}

We next prove the following lemma, identifying invariants of~\Cref{alg:bc}.

\begin{lemma}
~\Cref{alg:bc} satisfies the following invariants:

\newcounter{myinvariants}

\begin{list}{I\arabic{myinvariants}:}{\usecounter{myinvariants}\setlength{\rightmargin}{\leftmargini}}

    \item If $\lr{\lr{ts,v}_i,\sigma}$ (where $\sigma$ is a set of $f+1$ ready signatures) appears in $deliver_j$ for any processes $i,j$, then $\lr{ready,\lr{ts,v}_i}_k\in ready_k$ for a correct process $k$.\label{inv:deliverthenready}
    
    \item If $\lr{ready,\lr{ts,v}_i}_j\in ready_j$ for a correct process $j$, then $\lr{ts,v}_i\in echo_j$.\label{inv:readythenecho}
    
    \item If $\lr{ready,\lr{ts,v}_i}_j$ appears in $ready_j$ and $\lr{ready,\lr{ts,w}_i}_{j'}$ appears in $ready_{j'}$ for any two correct processes $j,j'$ then $v=w$.\label{inv:ready_agreement}

   \item If $\lr{\lr{ts,v}_i,\sigma}$ appears in $deliver_j$ and $\lr{\lr{ts,w}_i,\sigma}$ appears in $deliver_{j'}$ for any two correct processes $j,j'$ then $v=w$.\label{inv:deliver_agreement}
    
\end{list}
\end{lemma}

\begin{proof}
\newcounter{myinvariants2}

\begin{list}{I\arabic{myinvariants2}:}{\usecounter{myinvariants2}\setlength{\rightmargin}{\leftmargini}}

    \item Since $\lr{\lr{ts,v}_i,\sigma}$ appears in $deliver_j$ and it contains a set of $f+1$ signatures on $\lr{ready,\lr{ts,v}_i}$, there is at least one correct process $k$ that signed $\lr{ready,\lr{ts,v}_i}$ and added it to its ready register. By~\Cref{inv:remains}, it is not deleted from the register.

    \item Immediate from the code and~\Cref{inv:remains}.

    \item
    Since $\lr{ready,\lr{ts,v}_i0_j}$ appears in $ready_j$ and $j$ is correct, by~I\ref{inv:readythenecho} at least one correct process signed $\lr{ts,v}_i$ and added it to its echo register.
    Let $p_1$ be the first correct process to do so, and let $t_1$ be the moment of adding $\lr{ts,v}_i$ to $echo_{p_1}$ (see~\Cref{fig:conflicting_echoes} for illustration).
    By~\Cref{inv:remains}, it is not deleted from the register.
    Similarly, let $p_2$ be the first correct process to add $\lr{ts,w}_i$ to $echo_{p_2}$ at time $t_2$. WLOG, $t_1\geq t_2$.
    In addition, let $p_3$ be the first correct process to add $\lr{ready,\lr{ts,v}_i}$ to $ready_{p_3}$, and let $t_3$ be the moment of the addition. 
    By~I\ref{inv:readythenecho} it follows that $t_3>t_1$.
    By~\Cref{inv:remains}, the content of $echo_{p_2}$ and $ready_{p_3}$ is not deleted during the run.
    By the protocol, at some point in time between $t_1$ and $t_3$, $p_3$ executes line~\ref{line:conf1} and reads all echo registers.
    Let $t_1<t^*<t_3$ be the time when $p_3$ reads $echo_{p_2}$.
    Since $t_1\geq t_2$ we conclude that $t^*>t_2$.
    Since, $p_3$ does not see a conflicting value in $echo_{p_2}$, we get that $v=w$.
    
    \item 
    By~I\ref{inv:deliverthenready} at least one correct process $j$ signed $\lr{ready,\lr{ts,v}_i}$ and added it to $ready_j$ and at least one correct process $j'$ signed $\lr{ready,\lr{ts,w}_i}$ and added it to $ready_{j'}$. Thus, by~I\ref{inv:ready_agreement} $v=w$.

\end{list}
\end{proof}

\begin{figure}
    \centering
    \includegraphics[scale=0.7]{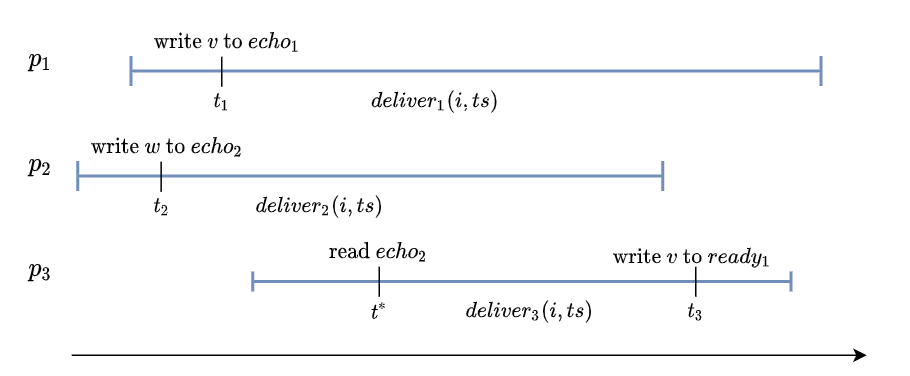}
    \caption{Concurrent deliver operations.}
    \label{fig:conflicting_echoes}
\end{figure}

Let us examine an execution $E$ of the algorithm.
Let $H$ be the history of $E$.
First, we define $H^c$ to be the history $H$ after removing any pending deliver operations and any pending broadcast operations that did not complete line~\ref{line:copy_deliver} (which is called from line~\ref{line:call_deliver}).
We define $H'$ to be an augmentation of $H^c|{correct}$ as follows.
For every Byzantine process $j$ and a value $v$ such that $v$ is returned by \emph{deliver$_i$(j,ts)} for at least one correct process $i$, we add to $H'$ a \emph{broadcast$_j$(ts,v)} operation that begins and ends immediately before the first correct process adds $\lr{\lr{ts,v}_j,\sigma}$ to its delivery register. Since at least one correct process adds this value at line~\ref{line:copy_deliver}, this moment is well-defined.
We construct a linearization $E'$ of $H'$ by defining the following linearization points:

\begin{itemize}

    \item Let $o$ be a \emph{broadcast$_i$(ts,v)} operation by a correct process $i$ that completed line~\ref{line:copy_deliver}. Note that by the code every completed \emph{broadcast} operation completes line~\ref{line:copy_deliver} exactly once, and operations that do not complete this line are removed from $H'$.
    The operation linearizes when $\lr{\lr{ts,v}_j,\sigma}$ is added for the first time to delivery register of a correct process, which occurs either when $i$ executes line~\ref{line:copy_deliver} or when another correct process executes line~\ref{line:add_to_deliver} beforehand. By the code, these lines are between the invocation and the return of the broadcast procedure.

    \item Let $o$ be a \emph{deliver$_i$(j,ts)} operation by a correct process $i$ that completes line~\ref{line:copy_deliver} and returns $v\neq\bot$ (note that by the code every completed \emph{deliver} operation that returns $v\neq\bot$ completes line~\ref{line:copy_deliver} exactly once).
    If $i$ finds $\lr{\lr{ts,v}_j,\sigma}$ for some value $v$ in some correct process' deliver register at line~\ref{line:complete_deliver}, then the operation linearizes when $i$ first reads $\lr{\lr{ts,v}_j,\sigma}$ from a correct process.
    Otherwise, it linearizes at line~\ref{line:copy_deliver} when $i$ copies the data to $deliver_i$.
    
    \item If $o$ is a completed \emph{deliver$_i$(j,ts)} operation by a correct process $i$ that returns $\bot$ it linearizes at the moment of its invocation.

    \item Every Byzantine \emph{broadcast$_j$(ts,v)} operation by process $j$ linearizes at the moment we added it.
\end{itemize}

In $H'$ there are no deliver operations by Byzantine processes. The following lemmas prove that $E'$, the linearization of $H'$, satisfies the sequential specification:
\begin{lemma}
For a given \emph{deliver(j,ts)} operation that returns $v\neq\bot$, there is at least one preceding broadcast operation in $E'$ of the form \emph{broadcast(ts,v)} invoked by process $j$.
\end{lemma}

\begin{proof}

Let $o$ be a \emph{deliver$_i$(j,ts)} operation invoked by a correct process $i$ that returns $v\neq\bot$.
Let $t$ be the time when $\lr{\lr{ts,v}_j,\sigma}$ is added for the first time to a delivery register of a correct process (where $\sigma$ contains $f+1$ ready signatures). 
If $j$ is correct then by~\Cref{inv:bc_occured} $j$ previously invoked \emph{broadcast(ts,v)} and that broadcast linearizes at time $t$.
If $j$ is Byzantine then \emph{broadcast(ts,v)} by process $j$ is added to $H'$ immediately before $t$. There are two options to the linearization point of $o$.
If $i$ finds $\lr{\lr{ts,v}_j,\sigma}$ in some correct process' deliver register at line~\ref{line:complete_deliver}, then $o$ linearizes when $i$ first reads $\lr{\lr{ts,v}_j,\sigma}$ from a correct process and thus it is after time $t$. Otherwise, it linearizes at line~\ref{line:copy_deliver} when $i$ copies the data to $deliver_i$, which is also no earlier than time $t$.
\end{proof}

\begin{lemma}
For a \emph{broadcast$_i$(ts,v)} in $E'$, there does not exist any \emph{broadcast$_i$(ts,w)} in $E'$ for $v\neq w$.
\end{lemma}
\begin{proof}
If $i$ is a correct process, the proof follows from~\Cref{assumption:onets}. If $i$ is Byzantine, \emph{broadcast$_i$(ts,v)} is added immediately before the first correct process adds $\lr{\lr{ts,v}_i,\sigma}$ to its delivery register. By~I\ref{inv:deliver_agreement}, no correct processes add $\lr{\lr{ts,w}_i,\sigma}$ to their delivery register for $v\neq w$ and \emph{broadcast$_i$(ts,w)} does not appear in $E'$.
\end{proof}

\begin{lemma}
For a given \emph{deliver(j,ts)} operation that returns $\bot$, there is no preceding broadcast operation in $H'$ of the form \emph{broadcast(ts,v)} invoked by process $j$, for $v\neq\bot$.
\end{lemma}
\begin{proof}
Let $o$ be a \emph{deliver(j,ts)} operation invoked by a correct process $i$ that returns $\bot$. Assume by way of contradiction that there is a preceding \emph{broadcast(ts,v)} operation in $H'$ invoked by process $j$, for $v\neq\bot$. By definition, the broadcast linearizes no later than the first adding of $\lr{\lr{ts,v}_j,\sigma}$ to a delivery register of a correct process. Thus, since $o$ linearizes at the moment of its invocation, it sees $\lr{\lr{ts,v}_j,\sigma}$ at some process' delivery register and returns $v\neq\bot$, in contradiction.
\end{proof}

Next, we prove $f$-resilience.

\begin{lemma}{(Liveness)}
Every correct process that invokes some operation eventually returns.
\end{lemma}
\begin{proof}
If a correct process $i$ invokes a deliver operation then by the code it returns in a constant time.
If it invokes \emph{broadcast(ts,v)}, it copies $\lr{ts,v}_i$ to $send_i$. By~\Cref{assumption:refresh}, all correct processes infinitely often call the reliable broadcast API and specifically the refresh procedure, see $\lr{ts,v}_i$ and copy it to their echo registers. As signatures are unforgable and $i$ is correct they do not find $\lr{ts,w}_i$ for any other $w\neq v$ in any other echo registers and copy a signed $\lr{ready,\lr{ts,w}_i}$ to their ready registers. By~I\ref{inv:remains}, eventually they all see $\lr{ready,\lr{ts,w}_i}$ in $f+1$ ready registers and copy $\lr{ts,w}_i$ to their deliver registers. Eventually $f+1$ correct processes have $\lr{ts,w}_i$ in their deliver registers, and since the signatures are valid, the check at line~\ref{line:complete_deliver} evaluates to true, and $i$ returns $v$ and finish the repeat loop.
\end{proof}

We conclude the following theorem:

\begin{theorem}
\Cref{alg:bc} implements an $f$-resilient Byzantine linearizable reliable broadcast object for any $f<\frac{n}{2}$.
\end{theorem}

\section{Byzantine Snapshot: Correctness}
\label{snapshot_appendix}

\begin{lemma}
\label{lemma:tschangedinupdate}
For a correct process $i$, at each point during an execution $collect_i[i]$ contains the value signed by $j$ with the highest timestamp until that point.
\end{lemma}
\begin{proof}
By induction on the execution; $collect_i[i]$ can change either at line~\ref{line:update_collect1} or at line~\ref{line:update_collect2}. If it changes at line~\ref{line:update_collect1}, $ts_i$ is increased and  $collect_i[i]$ contains the value with the highest timestamp. By induction, no signed value encountered at line~\ref{line:mono_check} has s timestamp higher than the one in $collect_i[i]$, so it is not updated at line~\ref{line:update_collect2}.
\end{proof}

\begin{lemma}
\label{lemma:mono_of_collect}
For a correct process $i$, $collect_i$ is monotonically increasing.
\end{lemma}
\begin{proof}
Let $j\in[n]$.
We prove that every time the value in $collect_i[j]$ is updated from $m$ to $m'$, it holds that $m'.ts>m.ts$. 
By the code $collect_i[j]$ changes either at line~\ref{line:update_collect1} or at line~\ref{line:update_collect2}. In both cases, the value in $collect_i[j]$ is signed by $j$.
If $collect_i[j]$ changes at line~\ref{line:update_collect2}, then monotonicity is immediate from the condition at line~\ref{line:mono_check}. Otherwise, it changes at line~\ref{line:update_collect1}, indicating that $i=j$ and monotonicity follows from~\Cref{lemma:tschangedinupdate}.
\end{proof}

\begin{lemclone}{lemma:seq}
The $i^{th}$ entry of the array returned by a \emph{snapshot} invocation in $E$ contains the value $v$ last updated by an \emph{update(v)} invoked by process $i$ in $E$, or its variable's initial value if no update was invoked.
\end{lemclone}
\begin{proof}
Let $v$ be the value in the $i^{th}$ entry of the array returned by a \emph{snapshot}, with a corresponding timestamp $ts_v$.
By the definition of $E$, \emph{update(v)} by process $i$ with timestamp $ts$ is linearized immediately before $ts_v\geq ts$.
If $i$ is correct and multiple update operations by $i$ are linearized at that point, then since $i$ invokes updates sequentially and by~\Cref{lemma:tschangedinupdate,lemma:mono_of_collect} their start times are ordered according to the increasing timestamps.
Thus, as updates are linearized by their start times, $v$ matches the value of the last update.
If $i$ is Byzantine, since we add updates only for values at the moment they are seen, $v$ must match the value of the last update.
Additionally, if $v$ is an initial value, then no updates were linearized before it in $E$.
\end{proof}

\begin{observation}
\label{lemma:snap_contains_collected}
For a snapshot operation invoked by a correct process $i$, let $c_i$ be the collected array at line~\ref{line:do_collect} and let $s$ be the return value. Then, $s\geq c_i$.
\end{observation}
\begin{proof}
Immediate from the condition at line~\ref{line:return_snapshot}.
\end{proof}

\begin{lemclone}{lemma:following_snaps}
If a snapshot operation invoked by a correct process $i$ with return value $s_i$ precedes a snapshot operation invoked by a correct process $j$ with return value $s_j$, then $s_i\leq s_j$.
\end{lemclone}
\begin{proof}
Assume $i$ invokes snapshot operation $snap_i$, which returns $s_i$ before $j$ invokes snapshot $snap_j$, returning $s_j$.
Let $c_1$ be the value of $collect_i$ that $j$ reads at line~\ref{line:do_collect} of $snap_j$ and let $c_2$ be the value it writes in $collect_j$ at line~\ref{line:updatedc}.
At the end of the last snapshot-aux in $snap_i$, $collected_i\geq s_i$ either because the return value is $collected_i$ (if snapshot-aux returns at line~\ref{line:return_snapaux}), or because $s_i$ is reflected in collect by the end of line~\ref{line:up_collect_w_save} if it is a savesnap returned at line~\ref{line:ret_arr1} or at line~\ref{line:ret_arr2}. Due to the monotonicity of collects (\Cref{lemma:mono_of_collect}), $s_i\leq c_1$.
Because $j$ reads $c_1$ when calculating $c_2$, $c_1\leq c_2$.
Finally, by~\Cref{lemma:snap_contains_collected}, $ c_2\leq s_j$ and by transitivity we get that $s_i\leq s_j$.
\end{proof}

\begin{lemclone}{lemma:snap_sees_update}
Let $s$ be the return value of a snapshot operation $snap_i$ invoked by a correct process $i$. Let $update_j(v)$ be an update operation invoked by a correct process $j$ that writes $\lr{ts,v}$ and completes before $snap_i$ starts. Then, $s[j].ts\geq ts$.
\end{lemclone}
\begin{proof}

Let $t_1$ be the time when $j$ completes line~\ref{line:update_collect1} in $update_j(v)$ and writes $\lr{ts,v}$. Let $t_2$ be the time when $i$ reads $collect_j[i]$ at line~\ref{line:do_collect} in $snap_i$.
By~\Cref{lemma:tschangedinupdate,lemma:mono_of_collect}, since $j$ is correct, it follows that $collect_j[j].ts\geq ts$ at time $t_2\geq t_1$. Thus, after line~\ref{line:updatedc} in $snap_i$ $collect_i[j].ts\geq ts$ and by~\Cref{lemma:snap_contains_collected}, $s[j].ts\geq ts$.


\end{proof}


\begin{invariant}
\label{inv_sup}
For any correct process $i$ that invokes snapshot-aux($k$), it holds that $collect_i$ is the supremum of the arrays in start message sent by processes in $senders_i$ from line~\ref{line:round_0_bc} and until the return value of snapshot-aux($k$) is determined at line~\ref{line:set_ret} or at line~\ref{line:return_snapaux}.
\end{invariant}
\begin{proof}
First, at line~\ref{line:round_0_bc} $senders_i$ contains $i$ itself, and $i$ sends exactly its $collect_i$ array.
The argument continues by induction on steps of snapshot-aux($k$).
Other than line~\ref{line:up_collect_w_save}, $collect_i$ and $senders_i$ change together:
Whenever $i$ receives a start message with an array $c$ from process $j$, it updates $collect_i$ with the higher-timestamped values found in $c$ and adds $j$ to $senders_i$ (lines~\ref{line:invoke_update_collect}--~\ref{line:up_senders_1}).

\end{proof}

\begin{definition}
\label{def:stability}
We say that the stability condition holds for a return value $s_1$ of snapshot-aux($k$) with a round $r$ and a set of processes $S$  if (1) $|S|\geq f+1$, (2) there is a set $S' \supseteq S$ so that for each $p \in S$ the union of all jsenders sets sent in $p$'s messages in rounds 1 to $r$ is $S'$, and (3) $s_1$ is the supremum of the collects sent in start messages of members of $S'$. 
\end{definition}

\begin{observation}
\label{newob}
If $s_1$ is returned from snapshot-aux($k$) by a correct process $i$, then $s_1$ satisfies the stability condition for some set $S$ in some round $r$.
\end{observation}
\begin{proof}
Consider two cases.
First, if $i$ returns $s_1$ at line~\ref{line:return_snapaux}, then the condition is satisfied for $s_1$ with the round $s$ that satisfies the condition at line~\ref{line:requestedcondition} and the set of $f+1$ processes for which the condition at line~\ref{line:requestedcondition} holds.
$S'$ is the set in $senders_i$ at the time the condition is satisfied. Since messages are delivered in order, we get that $S' \supseteq S$. 
Because the return value is $collect_i$, (3) follows from~\Cref{inv_sup}.

Second, if $i$ adopts a saved snapshot $s_1$ with a proof and returns at line~\ref{line:ret_arr1} or at line~\ref{line:ret_arr2}, then the proof contains $f+1$ messages from some round $r$ and corresponding start messages satisfying the stability condition.
\end{proof}




\begin{lemma}
\label{same_aux}
For a given $k$, Let $i,j$ be two correct processes that return $s_i,s_j$ from snapshot-aux($k$).
Then $s_i\leq s_j$ or $s_i>s_j$.

\end{lemma}
\begin{proof}

By~\Cref{newob}, $s_i$ satisfies the stability condition for some set $S_1$ in some round $r_1$. Let $S_1'$ be the set guaranteed from the definition.
Also by~\Cref{newob}, $s_j$ satisfies the stability condition and some set $S_2$ in some round $r_2$. 
Let $S_2'$ be the set guaranteed from the definition.

Since $|S_1|\geq f+1$ and $|S_2|\geq f+1$, there is at least one process $p\in S_1\cap S_2$.
Due to reliable broadcast, $p$ cannot equivocate with the set of processes $jsenders$ sent in each round of snapshot-aux($k$).

\begin{itemize}
    \item If $r_1=r_2$:
    By property (2) of~\Cref{def:stability} $S_1'=S_2'$, and by (3) $s_i=s_j$.

    \item If $r_1\neq r_2$:
    Assume WLOG $r_1<r_2$.
    Since the union of all jsenders sets sent in $p$'s messages in rounds 1 to $r_2$ is a superset of those sent in    rounds 1 to $r_1$, $S_2\supseteq S_1$ and then by (3) $s_j\geq s_i$.
    
\end{itemize}

    

    
    

\end{proof}

\begin{lemma}
\label{lemma:diff_an}
Let $i, j$ be two correct processes returning $s_i,s_j$ resp. from snapshot-aux with $auxnum=k$, such that $s_j>s_i$. 
Then when $i$ begins any snapshot-aux$_i(k')$ for $k'>k$, $collect_i>s_j$.
\end{lemma}

\begin{proof}

Since $j$ is correct, by~\Cref{newob}, $s_j$ satisfies the stability condition.
Let $t_1$ be a time when the condition is satisfied.
At time $t_1$, there is at least one correct process $l$ such that $collect_l\geq s_j$.
We show that either 
(1) $j$ does not return $s_j$ or 
(2) $i$ begins snapshot-aux$_i(k')$ with $collect_i>s_j$.
If $i$ begins snapshot-aux$_i(k')$ after $t_1$, then when it updates its collect at lines~\ref{line:do_col2}--\ref{line:do_collect2}, it reads the values in $collect_l$. By~\Cref{lemma:mono_of_collect}, $collect_l$ is greater than or equal to its value at time $t_1$. Thus, we get that $collect_i\geq collect_l\geq s_j$ and (2) holds. 
Otherwise, $i$ saves $s_i$ at line~\ref{line:updatesavesnap} before starting snapshot-aux($k'$), which is before time $t_1$. Between time $t_1$ and the time it returns $s_j$, j checks stored snapshots (at line~\ref{line:checkstored1}). When it does so, $j$ reads $s_i$, and since $s_j>s_i$ and $j$ returns the minimal array it sees, (1) holds. 

\end{proof}

\begin{lemma}
\label{lemma:snap_geq_collect}
If snapshot-aux$_i(k)$ of a correct process $i$ returns $s_i$, there is a correct process $j$ s.t. 
$j$ invoked snapshot-aux$_j(k)$ and $s_i\geq c_j$, where $c_j$ is the value of $collected_j$ after the collection at line~\ref{line:do_collect2} in snapshot-aux($k$) at $j$.
\end{lemma}
\begin{proof}
If snapshot-aux$_i(k)$ returns at line~\ref{line:return_snapaux}, then $i$ returns $collect_i$ and by~\Cref{lemma:mono_of_collect}, $s_i=collect_i$ is greater than or equal to its value after the collection at line~\ref{line:do_collect2} so the lemma holds with $i=j$.
Otherwise, snapshot-aux$_i(k)$ returns $s_i$ at line~\ref{line:ret_arr1} or at line~\ref{line:ret_arr2} and $s_i$ is an array saved in savesnap with a proof $\sigma$ signed by process $p$.
Since $i$ validates $s_i$, there was a round $r$ such that $|\{j |$ $seen_p[j][s]=senders_p\}|\geq f+1$. Thus, there was at least one correct process $j$ in this set. Since $j$ adds itself to $senders_j$ (\cref{line:add_to_starts_local}), $senders_j$ is broadcast by $j$ at every round (\cref{line:bc_rounds}), and it is the set added to $seen$, the array $c_j$ sent in $j$'s start message is reflected in $s_i$. This set is exactly the value of $collected_j$ after the collection at line~\ref{line:do_collect2} in snapshot-aux($k$) at $j$, and hence $s_i\geq c_j$.
\end{proof}

\begin{lemclone}{lemma:concurrent_snaps}
If two snapshot operations invoked by correct processes return $s_i$ and $s_j$, then $s_j\geq s_i$ or $s_j<s_i$. 
\end{lemclone}

\begin{proof}
By the code, $s_i$ is the return value of some snapshot-aux$_i(k_i)$ and $s_j$ is the return value of some snapshot-aux$_j(k_j)$.
WLOG, $k_i\geq k_j$. 

\begin{itemize}
    \item If $k_i=k_j$, the proof follows from~\Cref{same_aux}.


    \item If $k_i>k_j$:
    By~\Cref{lemma:snap_geq_collect}, there is a correct process $l$ that invoked snapshot-aux$_l(k_i)$, collected $c_l$ at line~\ref{line:do_collect2} of snapshot-aux$_l(k_i)$ (where $c_l$ is the value of $collected_l$ at that time), and $s_i\geq c_l$.
    Let $s_l$ be the return value of snapshot-aux$_l(k_j)$ (note that $l$ invokes snapshot-aux with increasing auxnums, so such a value exists). Consider two cases. First, if $s_j>s_l$, then by~\Cref{lemma:diff_an}, $s_j\leq c_l$. Thus, $s_j\leq c_l\leq s_i$ and the lemma follows.
    Otherwise, $s_j\leq s_l$.
    At the end of snapshot-aux$_l(k_j)$ $collected_l\geq s_l$ because either the return value is $collected_l$, or $s_l$ is reflected in collect by the end of line~\ref{line:up_collect_w_save}. Due to the monotonicity of collects (\Cref{lemma:mono_of_collect}), $s_l\leq c_l$.
    We conclude that $s_j\leq s_l\leq c_l \leq s_i$, as required.
    
\end{itemize}

\end{proof}

\begin{lemclone}{sees}
If update1 by process $i$ precedes update2 by process $j$ and a snapshot operation $snap$ by a correct process sees update2, then $snap$ sees update1 as well.
\end{lemclone}
\begin{proof}
Let $s$ be the return value of a snapshot that sees update2.
By~\Cref{newob}, $s$ is the supremum of $collect$ arrays sent at line~\ref{line:round_0_bc}.
If $s$ sees update2, by~\Cref{lemma:tschangedinupdate}, it means that $s$ reflects $collect_j$ after line~\ref{line:update_collect1} of update2. After, $j$ performed line~\ref{line:collect_in_up} and update1 was reflected in $collect_j$.
Hence, $s$ sees update1 as well.
\end{proof}

We now prove the liveness of our snapshot algorithm.

\begin{lemma}
\label{lemma:liveness_aux}
Every correct process that invokes snapshot-aux(auxnum) eventually returns.
\end{lemma}
\begin{proof}

Assum by induction on auxnum that all snapshot-aux instances with $k'<k$ (if any) have returned at all correct processes. Then, for auxnum=$k$, all correct processes initiate reliable broadcast instances and broadcast $\lr{0,c}$.
This is because all correct process invoke snapshot infinitely often.
Since all messages by correct processes are eventually delivered, they all eventually complete line~\ref{line:beginarounds} in each round.
Because $|senders|$ is bounded, eventually the $senders$ sets of all correct process stabilize, and due to reliable broadcast, they contain the same set of processes for all correct processes. Thus, there is be a round $r$ for which the condition at line~\ref{line:requestedcondition} is satisfied.
Therefore, at least one correct process returns from snapshot-aux at line~\ref{line:return_snapaux} (if it did not return sooner). Before returning, it updates its savesnap register at line~\ref{line:updatesavesnap}.
If it returns at line~\ref{line:ret_arr1} or at line~\ref{line:ret_arr2} it also updates its savesnap register at line~\ref{line:storesnap1}.
Every other correct process that has not yet returned from snapshot-aux will read the updated savesnap in the next while iteration and will return at line~\ref{line:ret_arr1}.
\end{proof}

\begin{lemclone}{liveness}
Every correct process that invokes some operation eventually returns.
\end{lemclone}
\begin{proof}
If a correct process $i$ invokes an update operation then by the code it returns in constant time.
If $i$ invokes a snapshot operation at time $t$, let $c$ be the collected array at line~\ref{line:do_collect}. Additionally, let $k$ be the maximum $auxnum$ of any snapshot-aux operation that was initiated by some process before time $t$. 
By~\Cref{lemma:liveness_aux}, all snapshot-aux invocations eventually return.
At snapshot-aux$(k+1)$, all correct processes see $c$ at lines~\ref{line:do_col2}--\ref{line:do_collect2} when they update their collect. Since the return value is the supremum of $f+1$ collect arrays, it is guaranteed that when $i$ executes snapshot-aux$(k+1)$, the returned value $res$ will satisfy $res\geq c$.
\end{proof}

\end{appendices}

\end{document}